\documentclass[11pt]{amsart}
\usepackage[letterpaper,margin=1in] {geometry}               
\usepackage{graphicx}
\usepackage{amssymb}
\usepackage{epstopdf}
\newtheorem{theorem}{Theorem}

\newtheorem{lemma}[theorem]{Lemma}

\title{Random cluster dynamics for the Ising model is rapidly mixing}\thanks{This work was partially 
supported by the EPSRC grant EP/N004221/1. 
This work was done (in part) while the authors were visiting the Simons Institute for the Theory of Computing.}
\author{Heng Guo \and Mark Jerrum}
\address{School of Mathematical Sciences\\
Queen Mary, University of London, Mile End Road, London E1 4NS, United Kingdom.}

\let\rho=\varrho

\def\Pr{\mathop{\textrm{Pr}}\nolimits}
\def\epsilon{\varepsilon}

\usepackage{hyperref}
\hypersetup{colorlinks=true,citecolor=red,linkcolor=blue, urlcolor=blue}

\newcommand{\trans}[4]{\ensuremath{\left[\begin{smallmatrix} #1 & #2 \\ #3 & #4 \end{smallmatrix}\right]}}

\begin{document}

\begin{abstract}
We show that the mixing time of Glauber (single edge update) dynamics for the
random cluster model at $q=2$ is bounded by a polynomial in the size of the underlying graph.
As a consequence, the Swendsen-Wang algorithm for the ferromagnetic Ising model at any temperature has the same polynomial mixing time bound.
\end{abstract}

\maketitle

\section{Introduction}
The Ising model is perhaps the best known model in statistical physics, 
and it has also been widely studied from an algorithmic perspective.  
An instance of the 
model is an undirected graph~$G$, together with a parameter $\beta>0$. 
A {\it configuration\/} of the model is an assignment $\sigma\in\{0,1\}^V$ of ``spins'' 
to the vertices of~$G$.  The {\it weight\/} $w(\sigma)$
of configuration~$\sigma$
is $\beta^{m(\sigma)}$ where $m(\sigma)$ is the number of monochromatic edges
(edges $\{i,j\}$ with $\sigma(i)=\sigma(j)$) in~$G$.  It is of importance to compute
the partition function of the system, which is the sum of weights $w(\sigma)$ 
over all configurations $\sigma\in\{0,1\}^V$.  

If $\beta<1$ then the system is antiferromagnetic, and the partition
function is computationally hard, even to approximate.  However, when $\beta>1$
the system is ferromagnetic, and the partition function can be approximated
in polynomial time to with any specified relative error~\cite{JSising}.  A direct approach
using Markov chain Monte Carlo (MCMC) on the spin configurations described above fails, 
as the spin model exhibits a phase transition for sufficiently large~$\beta$.  However,
there is an equivalent formulation of the Ising model in terms of ``even subgraphs''
which does form the basis for a successful application of MCMC, as was shown by 
Jerrum and Sinclair~\cite{JSising}.  (See Sections \ref{sec:IsingRC} and~\ref{sec:evensubgraph}
for details of the various models referred to in this introduction.)   

There is a third model which is equivalent to the Ising model in the sense of having the
same partition function up to an easily computable factor, namely the random cluster 
model introduced by Fortuin and Kasteleyn \cite{FortuinKasteleyn}.  In common with the even subgraphs model,
the configurations of the random cluster model are subsets of the edge set of~$G$.
However, the random cluster model is more tightly related to the Ising model;  in fact
a random Ising configuration can be obtained by colouring the connected components (clusters)
of a random cluster configuration independently and uniformly at random by $0$ and~$1$. 
Although we already have a polynomial-time algorithm for estimating the partition function
of the Ising model, it is natural to wonder about the mixing time of the Gibbs sampler
for random cluster configurations, which makes single edge-flip moves with Metropolis 
rejection probabilities.  For one thing, this dynamics may 
potentially mix faster than the standard dynamics for the even subgraphs model,
and the same is true with even greater force for the closely related Swendsen-Wang algorithm. 

Another reason for focusing on the random cluster model is that it extends the other 
two models in the following sense.
There is a generalisation of the Ising model 
to $q\geq2$ spins, known as the $q$-state Potts model, of which the Ising model
is the special case $q=2$.  Although the even subgraphs and spin formulations are defined
only for integer~$q$, the random cluster model makes sense for arbitrary positive real~$q$. 
Thus, by studying the dynamics of the the random cluster model at $q=2$, we may gain 
insight into the complexity of computing the partition function of the random cluster 
model at other values of $q$, particularly (for reasons that will be explained presently)
in the range $0\leq q<2$.  Stated in other terms, we would hope to gain information
about the complexity of approximating the Tutte polynomial $T(G;x,y)$ 
in the region $0\leq(x-1)(y-1)<2$,
and $x,y\geq1$, about which nothing is currently known except for the point $x=y=1$
and the (trivial) hyperbola $(x-1)(y-1)=1$ \cite{Tutte1,Tutte2}. 

In this paper we prove for the first time that the Gibbs sampler (single edge-flip 
dynamics) for the random cluster model on an 
arbitrary graph mixes in time polynomial in $n=|V(G)|$, the number of vertices
of~$G$. (See Theorem \ref{thm:RC:fast}.)  
One main tool is the well known canonical paths technique for bounding mixing
time via a parameter known as congestion (in the form presented by Sinclair~\cite{Sin92}, 
building on work of Diaconis and Stroock~\cite{DS93}).
Another tool is a coupling between random cluster and even subgraph configurations
discovered by Grimmett and Janson~\cite{GrimmettJanson}.  The existence of this 
coupling invites us to bound the 
congestion of the edge-flip dynamics on random cluster configurations in terms 
of the known bounds on congestion for the edge-flip dynamics on (augmented)
even subgraph
configurations, established by Jerrum and Sinclair~\cite{JSising}.  Unfortunately, this
translation between the models cannot be handled by existing comparison techniques \cite{DS93,comparison},
and an extension of comparison methods to the current situation is a contribution 
of the paper, and one that may find application elsewhere.

The Swendsen-Wang algorithm \cite{SwendsenWang} is widely considered to be
an efficient method for sampling random cluster configurations (and 
Ising spin configurations) in practice.  Ullrich has shown~\cite{Ullrich:Comparison} 
that the Swendsen-Wang 
dynamics mixes at least as fast as the edge-flip dynamics, so our result provides 
the first polynomial upper bound on mixing time of the Swendsen-Wang algorithm.
This provides a partial answer to a problem that has been open since around 
1990, when the Ising model was first studied from a complexity-theoretic perspective.
The answer is partial in the sense that the exponent in the bound we derive
here is likely to be well above the true answer, and is certainly too high
to be of practical interest. 
Hopefully, the result presented here may be the first step on the road to 
a practically useful upper bound on the Swendsen-Wang dynamics.  

Since the random cluster model is defined for all positive real~$q$, it is 
natural to speculate on the mixing time of the Glauber dynamics when $q\not=2$.
For $q>2$, the mixing time cannot be polynomial in general, owing to a first-order
phase transition of the model on the complete graph (the ``mean-field'' situation)
identified by Bollob\'as, Grimmett and Janson~\cite{BollobasGrimmettJanson}.
This phase transition is a barrier to rapid mixing when $q>2$, as shown 
by Gore and Jerrum when $q$ is an integer~\cite{GoreJerrum}, 
and by Blanca and Sinclair~\cite{BlancaSinclair:Meanfield}
for general $q>2$.  For $q$ sufficiently large, Borgs et al.~\cite{torpid}
prove exponential time mixing even in the physically important case of the two-dimensional
lattice.  In fact, there is no polynomial-time algorithm of any sort for 
evaluating the partition function of the random cluster model on general graphs
when $q>2$, unless there is an FPRAS for counting independent sets
in a bipartite graph~\cite{GoldbergJerrum:Potts}.  In contrast,
in the range $0\leq q\leq2$ there
is no known barrier to rapid mixing, and there is cause to be optimistic, particularly 
in the range $1<q<2$, in which the random cluster model is monotonic.   

In this version of the paper, we do not try too hard to optimise the exponent 
in the mixing-time bound, as the method we employ is unlikely to
get close to the true answer.

\section{Ising and Random Cluster model}\label{sec:IsingRC}
The ferromagnetic Ising model on a graph $G=(V,E)$ with parameter $\beta>1$ is defined by the following: 
for any $\sigma\in\{0,1\}^V$, the probability of being in configuration~$\sigma$ is
\begin{align}
  \pi(\sigma)=\frac{\beta^{m(\sigma)}}{Z_{Ising}(\beta)},
  \label{eqn:Ising}
\end{align}
where $m(\sigma)$ is the number of mono-chromatic edges in $\sigma$, and its normalizing factor, the so-called partition function, is defined as
\begin{align*}
  Z_{Ising}(\beta)=\sum_{\sigma\in\{0,1\}^V}\beta^{m(\sigma)}.
\end{align*}
The random cluster model with parameters $(p,q)$ is defined on subsets of edges $S\subseteq E$ such that
\begin{align}
  \pi_{RC}(S) \propto p^{|S|}(1-p)^{|E\backslash S|}q^{\kappa(S)},
  \label{eqn:RC}
\end{align}
and its partition function is
\begin{align*}
  Z_{RC}(p,q)=\sum_{S\subseteq E} p^{|S|}(1-p)^{|E\backslash S|}q^{\kappa(S)}.
\end{align*}
Denote this measure by $\pi_{RC;p,q}(\cdot)$ or simply $\pi_{RC}(\cdot)$ when there is no confusion.
We use $\Omega$ throughout this article to denote the state space of random cluster models, namely $\{0,1\}^{E}$.
It is well known that, for $q=2$ and $p=1-\frac{1}{\beta}$, the random cluster model is equivalent to the Ising model
in the sense that their partition functions are equal up to some easily computable factor (see \eqref{eqn:equivalence}).
The random cluster model was introduced by Fortuin and Kastelyn~\cite{FortuinKasteleyn}, who also described its
relationship to the Ising model.  The connection between the two models was further elucidated 
by Edwards and Sokal~\cite{EdwardsSokal}.

The (lazy) single bond flip dynamics $P_{RC}$ is defined as follows based on the Metropolis filter.
\begin{align}
  P_{RC}(x,y)=
  \begin{cases}
    \frac{1}{2m} \min\left\{1,\frac{\pi_{RC}(y)}{\pi_{RC}(x)}\right\} & \textrm{if } |x\oplus y| = 1 ;\\
    1-\frac{1}{2m}\sum_{e\in E} \min\left\{1,\frac{\pi_{RC}(x\oplus \{e\})}{\pi_{RC}(x)}\right\}& \textrm{if } x=y ;\\
    0 & \textrm{otherwise,}\\
  \end{cases}
  \label{eqn:SB:RC}
\end{align}
where $x,y\in\Omega$.
It is not hard to see, for example, by checking the detailed balance condition, 
that $\pi_{RC}(\cdot)$ is the stationary distribution of $P_{RC}$.
Note that the Markov chain is {\it lazy}, i.e., it remains 
at its current state with probability at least~$\frac12$.  This eliminates the possibility of the 
transition matrix having $P$ having negative eigenvalues, and simplifies the analysis later.

For a Markov chain with transition matrix $P$ and stationary distribution $\pi$, we are interested in its 
{\it mixing time}, that is, how fast it converges to the stationary distribution, defined as follows:
\begin{align}
  \tau_{\epsilon}(P):=\min\left\{t:\max_{x\in\Omega}||P^t(x,\cdot)-\pi||\le \epsilon\right\},
  \label{eqn:mixingtime}
\end{align}
where $||\cdot||$ is the total variation distance, namely
\begin{align*}
  ||\pi-\pi'||=\frac{1}{2}\sum_{x\in\Omega}|\pi(x)-\pi'(x)|.
\end{align*}

Canonical paths are a useful technique to bound the mixing time of Markov chains,
introduced by Jerrum and Sinclair \cite{SJ,JSperm}.
Let $\Gamma=\{\gamma_{xy}:x,y\in\Omega\}$ be a collection of paths, 
where $\gamma_{xy}$ is a ``canonical'' path from $x$ to $y$ using transitions of the Markov chain.
The {\it congestion\/} $\rho(\Gamma)$ associated with these paths is 
\begin{align}
  \rho(\Gamma):=\max_{(z,z')\in\Omega^2, P(z,z')>0}\,\frac{L}{\pi(z)P(z,z')}\,\sum_{\substack{x,y\in\Omega^2\\ \gamma_{xy}\ni(z,z')}}\pi(x)\pi(y),
  \label{eqn:congestion}
\end{align}
where $L=L(\Gamma)$ denotes the maximum length of paths in $\Gamma$. 

A more general technique is provided by the flow formulation for congestion.
A flow $\Gamma$ is a collection of paths, and each path $\gamma\in\Gamma$ is assigned a weight $wt(\gamma)$, such that
\begin{align}
  \sum_{\gamma \textrm{ is from } x \textrm{ to } y} wt(\gamma)=\pi(x)\pi(y).
  \label{eqn:flow:requirement}
\end{align}
The congestion of $\Gamma$ is defined as
\begin{align}
  \rho(\Gamma):=\max_{(z,z')\in\Omega^2, P(z,z')>0}\,\frac{L}{\pi(z)P(z,z')}\,\sum_{\gamma\in\Gamma,\, (z,z')\in\gamma}wt(\gamma).
  \label{eqn:congestion:flow}
\end{align}
The canonical paths are just a flow where for each pair $(x,y)$ there is only one path with positive weight.

It is standard that the mixing time of a Markov chain $P$ can be bounded by the congestion of 
any flow~$\Gamma$~\cite{Sin92}.
\begin{theorem}
  For a lazy, ergodic, reversible Markov chain $P$ and any initial state $x_0\in\Omega$,
  \begin{align*}
    \tau_{\varepsilon}(P)\le\rho(\Gamma)(\ln\pi(x_0)^{-1}+\ln\varepsilon^{-1}).
  \end{align*}  
  \label{thm:mixingtime:congestion}\vspace{-3\belowdisplayskip}
\end{theorem}
Our goal is to bound $\tau_{\varepsilon}(P_{RC})$.
We can choose the initial state to be the empty set of edegs, 
which has weight $\pi(\emptyset)=\frac{(1-p)^{|E|}2^{|V|}}{Z_{RC}}$.
Also for $\beta=\frac{1}{1-p}$ we have
$Z_{RC}(p,2)= \beta^{-|E|}Z_{Ising}(\beta) \le 2^{|V|} $, and therefore $\pi(\emptyset)\ge (1-p)^{|E|}$.
Hence, $\ln\pi(x_0)^{-1}\le m\ln(1-p)^{-1}$.
The main task is to design a good flow $\Gamma_{RC}$ so that $\rho(\Gamma_{RC})$ is bounded by a polynomial.

\section{Random Even Subgraphs} \label{sec:evensubgraph}

There is yet another formalism of the Ising model, 
that is, the so-called ``high-temperature expansion'' or even subgraphs model.
We still pick a subset of edges $S\subseteq E$ but with the further restriction 
that every vertex in the induced subgraph $(V,S)$ has even degree.
Denote by $\Omega_{even}(G)$ the state space of all such even subgraphs of $G$.
We usually simply write $\Omega_{even}$ when there is no confusion.
In this even subgraphs model we want to sample from $\Omega_{even}$ with parameter $p\le 1/2$,
so that edges are more inclined to be ``out'' than ``in''.
That is, for any $S\in\Omega_{even}$,
\begin{align}
  \pi(S)\propto p^{|S|}(1-p)^{|E\backslash S|}
  \label{eqn:even}
\end{align}
and
\begin{align*}
  Z_{even}(p)=\sum_{S\in\Omega_{even}} p^{|S|}(1-p)^{|E\backslash S|}.
\end{align*}
Distributions \eqref{eqn:Ising}, \eqref{eqn:RC}, and \eqref{eqn:even} have in fact the same partition function,
up to certain scaling factors:
\begin{align}
  Z_{Ising}(\beta)=\beta^{|E|}Z_{RC}\left(1-\frac{1}{\beta},2\right)=2^{|V|}\beta^{|E|}Z_{even}\left(\frac{1}{2}\left(1-\frac{1}{\beta}\right)\right).
  \label{eqn:equivalence}
\end{align}
The first equivalence is well-known, cf.\ \cite{Grimmett:book}.
The second one is also a classical result, cf.\ \cite{Wae41}.
More detailed explanations can be found in Appendix~\ref{sec:equivalence}.

Grimmett and Janson \cite[Thm~3.5]{GrimmettJanson} discovered the following coupling between even subgraphs and random cluster configurations.
Take a random even subgraph $S$ from distribution \eqref{eqn:even} with parameter $p\le 1/2$.
Then we add each edge $e\not\in S$ independently with probability $\frac{p}{1-p}$ to get $R$.

\begin{theorem}{\cite[Thm~3.5]{GrimmettJanson}}
  The subgraph $R$ is a random cluster configuration, that is, it satisfies \eqref{eqn:RC} with parameters $(2p,2)$.
  \label{thm:GJ:coupling}
\end{theorem}
For completeness we give a proof of Theorem \ref{thm:GJ:coupling}.
\begin{proof}
  The number of even subgraphs of a (not necessarily simple) graph $G=(V,E)$ is well known to be
  \begin{align}
    |\Omega_{even}(G)|=2^{|E|-|V|+\kappa(G)}
    \label{eqn:even:number}
  \end{align}
  where $\kappa(G)$ is the number of connected components of $G$.
  
  For each $r\subseteq E$,
  \begin{align*}
    \Pr(R=r) & \propto \sum_{s\subseteq r, s \textrm{ even}} 
    \left( \frac{p}{1-p} \right)^{|s|}
    \left( \frac{p}{1-p} \right)^{|r\backslash s|}
    \left( \frac{1-2p}{1-p} \right)^{|E\backslash r|}\\
    & \propto p^{|r|}(1-2p)^{|E\backslash r|}N(r),
  \end{align*}
  where $N(r)$ is the number of even subgraphs of $(V,r)$.
  By \eqref{eqn:even:number}, $N(r)=2^{|r|-|V|+\kappa(r)}$.
  Hence,
  \begin{align*}
    \Pr(R=r) & \propto (2p)^{|r|}(1-2p)^{|E\backslash r|}2^{\kappa(r)}.\qedhere
  \end{align*}
\end{proof}

However, it is not clear how to sample from $\Omega_{even}$ with edge weights directly in an efficient way,
partly because of the rigid structure of the all even requirement.
On the other hand, Jerrum and Sinclair \cite{JSising} designed a Markov chain to do so by moving among all subgraphs, 
but with each odd degree vertex incurring a penalty.
Note that the Jerrum-Sinclair Markov chain together with the Grimmett-Janson coupling (Theorem \ref{thm:GJ:coupling})
yields an efficient sampler for random cluster models and Ising configurations.
It is more straightforward and efficient than the one given by Randall and Wilson~\cite{RW99},
which also uses the Jerrum-Sinclair chain.

A slightly simpler Markov chain is to move between even subgraphs and near-even subgraphs, 
for which we allow exactly two odd degree vertices (or ``holes'').
This is the so-called ``worm'' process, introduced by Prokof'ev and Svistunov \cite{PSworm}.

Let $\Omega_k$ be the collection of subgraphs where $k$ many vertices have odd degrees.
Then $\Omega_0=\Omega_{even}$ and the state space $\Omega_{worm}$ of the ``worm'' process is $\Omega_{worm}:=\Omega_0\cup\Omega_2$.
For each pair of vertices $(u,v)$ such that $u\neq v$, denote by $\Omega(u,v)$
the set of subgraphs of $G$ in which $u$ and $v$ have odd degrees and all other vertices are even.
Then
\begin{align*}
  \Omega_2=\bigcup_{u,v\in V}\Omega(u,v).
\end{align*}
For a subset of edges $S\subseteq E$, let $w_p(S):=p^{|S|}(1-p)^{|E\backslash S|}$.
We give a penalty of $n^{-2}$ to each near-even subgraph:
\begin{align}
  w_{worm}(S):=
  \begin{cases}
    w_p(S) & \textnormal{ if }S\in\Omega_0;\\
    n^{-2}w_p(S) & \textnormal{ if }S\in\Omega_2;\\
    0 & \textnormal{ otherwise.}
  \end{cases} 
  \label{eqn:worm:weight}
\end{align}
The ``worm'' measure is defined as the following:
\begin{align}
  \pi_{worm}(S):=
  \begin{cases}
    \frac{w_{worm}(S)}{Z_{worm}(p)} & \textrm{ if }S\in\Omega_{worm}; \\
    0 & \textrm{ otherwise,}
  \end{cases}
  \label{eqn:worm:measure}
\end{align}
where $Z_{worm}(p)=\sum_{S\in\Omega_{worm}}w_{worm}(S)$.

The winding idea of \cite{JSising} provides a way to design canonical paths between states in $\Omega_{worm}$ with low congestion.
We will not need to analyze it in full detail for the worm process.
Instead, we only care about paths from an even subgraph to another.

\begin{theorem}
  There is a collection of paths 
  \begin{align*}
    \Gamma_{worm}=\{\gamma_{xy}\,|\,x,y\in\Omega_0\}
  \end{align*}
  such that $wt(\gamma_{xy})=\pi_{even}(x)\pi_{even}(y)$,
  and for any $\gamma\in\Gamma_{worm}$ and any state $w\in\gamma$, $w\in\Omega_{worm}$.
  Each state $w$ appears at most once in $\gamma$ and $L(\Gamma_{worm})\le m$.
  Moreover, for any transition $(w,w')$ where $w'=w\oplus\{e\}$,
  \begin{align*}
    \sum_{\gamma\ni(w,w')}wt(\gamma)\le n^4\pi_{worm}(w).
  \end{align*} 
  In the special case $w'=w\cup\{e\}$ for some $e\not\in w$, we have the additional bound
  \begin{align*}
    \sum_{\gamma\ni(w,w')}wt(\gamma)\le n^4\pi_{worm}(w)\frac{p}{1-p}.
  \end{align*}
  \label{thm:path:worm}
\end{theorem}

Note that $\Gamma_{worm}$ is not a complete collection of canonical paths for $\pi_{worm}(\cdot)$.
The proof of Theorem $\ref{thm:path:worm}$ is an adaptation of \cite{JSising}
and is given in Appendix~\ref{sec:worm}.  Note that Collevecchio et al.~\cite{Collevecchio:worm}
give an analysis of a complete set of canonical paths for the worm process, but their 
result does not quite fit our situation.  

Since paths in $\Gamma_{worm}$ go through $\Omega_{worm}$ instead of $\Omega_{even}$,
we need to extend Theorem \ref{thm:GJ:coupling} to $\Omega_{worm}$.
It will no longer be exact.

Take a random subgraph $S$ from distribution~\eqref{eqn:worm:measure} with parameter $p\le 1/2$.
Again we add each edge $e\not\in S$ independently with probability $\frac{p}{1-p}$ to get~$R$.
Call this measure $\widehat{\pi}(\cdot)$.

\begin{lemma}
  For any $R\subseteq E$,
  \begin{align*}
    \frac{\widehat{\pi}(R)}{\pi_{RC;2p,2}(R)}\le \frac{3}{2}.
  \end{align*}  
  \label{lem:worm:coupling}
\end{lemma}
\begin{proof}
  Similarly to the proof of Theorem \ref{thm:GJ:coupling},
  it is not hard to see that 
  \begin{align*}
    \widehat{\pi}(R) \propto p^{|R|}(1-2p)^{|E\backslash R|}(N(R)+n^{-2}N'(R)),
  \end{align*}
  where $N(R)$, as before, is the number of even subgraphs of $(V,R)$, and
  $N'(R)$ is the number of subgraphs of $R$ that belong to $\Omega_2$.
  Note that for each near-even subgraph there is a penalty of $n^{-2}$ for its weight (see \eqref{eqn:worm:weight}).
  We use \eqref{eqn:even:number} to count the number of even subgraphs of $R$, which is $2^{|R|-|V|+\kappa(R)}$.
  
  Let $\Omega_R(u,v)$ be the set of near even subgraphs of $R$ with holes $u$ and $v$.
  If $u,v$ are in different connected components of $(V,R)$,
  then there is no possible such subgraph and $|\Omega_R(u,v)|=0$.
  Otherwise $u,v$ are in the same component of $(V,R)$, 
  and we can add an extra edge $(u,v)$ to $R$ to get a graph $R'$.
  Applying \eqref{eqn:even:number} to $R'$ we get that
  \begin{align*}
    N(R')= 2^{|E|+1-|V|+\kappa(R)} = N(R)+|\Omega_R(u,v)|.
  \end{align*}
  The second equality is because each even subgraph of $R'$ either uses the new edge $(u,v)$ or not.
  If it uses $(u,v)$, then it is an even subgraph of $R$.
  Otherwise it is a near even subgraph of $R$ with holes $u$ and $v$.
  Hence,
  \begin{align*}
    |\Omega_R(u,v)|=2^{|E|-|V|+\kappa(R)},
  \end{align*}
  as $N(R)=2^{|E|-|V|+\kappa(R)}$.
  
  Let $c(R)$ be the number of pairs of vertices from every component of $(V,R)$.
  That is,
  \begin{align}
    c(R):=\sum_{i=1}^{\kappa(R)}{n_i \choose 2},
    \label{eqn:c(R)}
  \end{align}
  where $n_i$ is the size of the $i$th component of $(V,R)$ with the convention that ${1\choose 2}=0$.
  Then we have that
  \begin{align*}
    N'(R)=2^{|R|-|V|+\kappa(R)}c(R),
  \end{align*}
  and
  \begin{align*}
    \widehat{\pi}(R) \propto (2p)^{|R|}(1-2p)^{|E\backslash R|}2^{\kappa(R)}\left( 1 + \frac{c(R)}{n^2}\right).
  \end{align*}
  The lemma follows by noticing that $0\le c(R)\le \frac{n(n-1)}{2}$.
\end{proof}

\section{Lifting Canonical Paths}

Let $p\le 1/2$ be the parameter of the even subgraph and the worm measure.
Let $\Gamma_{worm}$ be the collection of paths as in Theorem \ref{thm:path:worm}.
We will use Lemma \ref{lem:worm:coupling} to lift $\Gamma_{worm}$ 
to a flow $\Gamma_{RC}$ for $P_{RC}$, the single edge-flip Markov chain for the random cluster model with parameter $2p$.

We first construct a flow $\Gamma_{RC}'$ from $\Gamma_{worm}$.
Let $\gamma=\{w_0,w_1,\cdots,w_\ell\}$ be a path in $\Gamma_{worm}$ where $w_0,w_\ell\in\Omega_0$, and $\ell\le L(\Gamma_{worm})$.
We lift $\gamma$ to a flow (random path) as follows.
First we add each edge $e\not\in w_0$ with probability $p'=\frac{p}{1-p}$ independently as in Lemma \ref{lem:worm:coupling}, to obtain the starting state~$Z_0$ of the path. 
In other words, letting
\begin{align*}
  \delta(w,z):=(p')^{z\backslash w}(1-p')^{E\backslash z},
\end{align*}
for subsets of edges $w\subseteq z\subseteq E$,
we draw a superset $Z_0$ of $w_0$ such that $\Pr(Z_0=z)=\delta(w_0,z)$ for any $z\supseteq w_0$.
Note that
\begin{align*}
  \pi_{RC}(z)=\sum_{w\subseteq z, w\in\Omega_0}\pi_{even}(w)\delta(w,z)
\end{align*}
by Theorem \ref{thm:GJ:coupling}, and
\begin{align*}
  \widehat{\pi}(z)=\sum_{w\subseteq z, w\in\Omega_{worm}}\pi_{worm}(w)\delta(w,z)
\end{align*}
by definition.

We construct $Z_1,\cdots,Z_\ell$ inductively.
Given $Z_{k-1}$ for $1\le k\le \ell$, we construct $Z_k$ by mimicking the transition from $w_{k-1}$ to $w_k$ 
while ensuring that 
\begin{align*}
  \Pr_\gamma(Z_k=z)=\delta(w_k,z),
\end{align*}
for any $z\supseteq w_k$ at the same time.
Here the subscript $\gamma$ emphasises that probabilities are with respect to a
fixed path~$\gamma$.
By induction hypothesis, $\Pr_\gamma(Z_{k-1}=z)=\delta(w_{k-1},z)$ for any $z\supseteq w_{k-1}$.
For $Z_k$, there are two cases:
\begin{itemize}
  \item If $w_k=w_{k-1}\cup\{e\}$ for some edge $e\not\in w_{k-1}$, 
    then let $Z_k=Z_{k-1}\cup\{e\}$.
    We have that 
    \begin{align*}
      \Pr_\gamma(Z_{k}=z) & = \Pr_\gamma(Z_{k-1}=z) + \Pr_\gamma(Z_{k-1}=z\backslash\{e\}) \\
      & = \delta(w_{k-1},z) + \delta(w_{k-1},z\backslash\{e\}) \\
      & = \delta(w_k,z)p' + \delta(w_k,z)(1-p') = \delta(w_k,z),
    \end{align*}
    for any $z\supseteq w_k$.
  \item If $w_k=w_{k-1}\backslash\{e\}$ for some edge $e\in w_{k-1}$, 
    then let $Z_k=Z_{k-1}$ with probability $p'$ and $Z_k=Z_{k-1}\backslash\{e\}$ with probability $1-p'$.
    For any $z\supseteq w_k$ such that $e\in z$,
    \begin{align*}
      \Pr_\gamma(Z_{k}=z) & = \Pr_\gamma(Z_{k-1}=z)p' = \delta(w_{k-1},z)p' = \delta(w_k,z),
    \end{align*}
    and for any $z\supseteq w_k$ such that $e\not\in z$,
    \begin{align*}
      \Pr_\gamma(Z_{k}=z) & = \Pr_\gamma(Z_{k-1}=z\cup\{e\})(1-p') = \delta(w_{k-1},z\cup\{e\})(1-p') = \delta(w_k,z).
    \end{align*}
\end{itemize}

Given $\gamma$, the random flow path $Z = \{Z_0,Z_1,\cdots,Z_\ell\}$ is constructed as above.
For a particular flow path $\zeta = \{z_0,z_1,\cdots,z_\ell\}$ lifted from some path~$\gamma$, 
assign its weight to be
\begin{align*}
  wt(\zeta)=\sum_{\gamma\in\Gamma_{worm}}wt(\gamma)\Pr_\gamma(Z=\zeta).
\end{align*}
This finishes the construction of $\Gamma_{RC}'$.

However, $\Gamma_{RC}'$ is not a valid flow for $\pi_{RC}(\cdot)$.
If we randomly draw a path from $\Gamma_{RC}'$ according to $wt(\cdot)$,
then $Z_0$ and $Z_\ell$ both are distributed according to $\pi_{RC}(\cdot)$.
However, $Z_\ell$ is correlated with $Z_0$ and is not distributed correctly conditional on $Z_0$.

We resolve this issue next by constructing $\Gamma_{RC}$.
Given $\gamma\in\Gamma_{worm}$ with length $\ell$, we construct $Z_0,\cdots,Z_\ell$ the same as in $\Gamma_{RC}'$.
To repair the distribution of $Z_\ell$, we append further transitions to re-randomize edges that are not in $w_\ell$.
More precisely, let $\{e_1,e_2,\cdots,e_{k}\}$ be the edges that are not in $w_\ell$ where $k=|E\backslash w_\ell|$.
Given $Z_{\ell+i-1}$ for $1\le i\le k$, let $Z_{\ell+i-1}'=Z_{\ell+i-1}\backslash\{e_i\}$.
Then $Z_{\ell+i}=Z_{\ell+i-1}'$ with probability $1-p'$ and $Z_{\ell+i}=Z_{\ell+i-1}'\cup\{e_i\}$ with probability $p'$.
As in $\Gamma_{RC}'$, for a particular flow path $\zeta = \{z_0,z_1,\cdots,z_{\ell+k}\}$ lifted from some~$\gamma$, 
its weight is defined to be
\begin{align*}
  wt(\zeta)=\sum_{\gamma\in\Gamma_{worm}}wt(\gamma)\textnormal{Pr}_{\gamma}(Z=\zeta).
\end{align*}
This finishes the construction of $\Gamma_{RC}$.
The longest path in $\Gamma_{RC}$ has length at most $L(\Gamma_{worm})+m$, that is, $L(\Gamma_{RC})\le L(\Gamma_{worm})+m\le 2m$.

Fix a path $\gamma=\{w_0,w_1,\cdots,w_\ell\}$.
For any $0\le i\le \ell$ and $z\supseteq w_i$, we have $\textnormal{Pr}_{\gamma}(Z_i=z)=\delta(w_i,z)$,
because of the construction of $\Gamma_{RC}'$.
Moreover, for any $1\le i\le |E\backslash w_\ell|$ and $z\supseteq w_\ell$, we have 
$\textnormal{Pr}_{\gamma}(Z_{\ell+i}=z)=\delta(w_\ell,z)$.
This can be shown by inductively going through the construction above.
The re-randomization does not change the marginal distribution but removes the correlation between $Z_0$ 
and $Z_{\ell'}$, where $\ell'=\ell+|E\backslash w_\ell|$ (conditional on $\gamma$).

The flow $\Gamma_{RC}$ is valid for $\pi_{RC}(\cdot)$.
We verify \eqref{eqn:flow:requirement} as follows:
\begin{align*}
  \sum_{\zeta \textnormal{ is from $x$ to $y$}}wt(\zeta) & = \sum_{\substack{w\subseteq x,\ w'\subseteq y\\w,w'\in\Omega_{0}}} \,\,
  \sum_{\gamma \textnormal{ is from $w$ to $w'$}} wt(\gamma) \textnormal{Pr}_{\gamma}(Z_0=x,Z_{\ell'}=y)\\
  & = \sum_{\substack{w\subseteq x,\ w'\subseteq y\\w,w'\in\Omega_{0}}} \,\,
  \sum_{\gamma \textnormal{ is from $w$ to $w'$}} wt(\gamma) \textnormal{Pr}_{\gamma}(Z_0=x)\textnormal{Pr}_{\gamma}(Z_{\ell'}=y)\\
  & = \sum_{\substack{w\subseteq x,\ w'\subseteq y\\w,w'\in\Omega_{0}}} \,\,
  \sum_{\gamma \textnormal{ is from $w$ to $w'$}} wt(\gamma) \delta(w,x)\delta(w',y)\\ 
  & = \sum_{\substack{w\subseteq x,\ w'\subseteq y\\w,w'\in\Omega_{0}}} 
  \pi_{even}(w)\pi_{even}(w') \delta(w,x)\delta(w',y)\\
  & = \left( \sum_{w \subseteq x,\ w \in\Omega_{0}} \pi_{even}(w) \delta(w,x)  \right)
      \left( \sum_{w'\subseteq y,\ w'\in\Omega_{0}} \pi_{even}(w')\delta(w',y) \right)\\
  & = \pi_{RC}(x)\pi_{RC}(y),
\end{align*}
where in the last step we use Theorem \ref{thm:GJ:coupling}.

\begin{lemma}
  Let $2p\le 1$ be the parameter for the random cluaster model.
  \begin{enumerate}
    \item For a transition $(z,z')$ where $z'=z\cup\{e\}$ for some $e\not\in z$, 
      \begin{align*}
        \sum_{\zeta\in\Gamma_{RC},\ \zeta\ni(z,z')}wt(\zeta) \le \frac{p}{1-p}\cdot 2n^4\pi_{RC}(z).
      \end{align*}
    \item For a transition $(z,z')$ where $z'=z\backslash\{e\}$ for some $e\in z$, 
      \begin{align*}
        \sum_{\zeta\in\Gamma_{RC},\ \zeta\ni(z,z')}wt(\zeta) \le \frac{1-2p}{1-p}\cdot 2n^4\pi_{RC}(z).
      \end{align*}      
    \item For a transition $(z,z)$, 
      \begin{align*}
        \sum_{\zeta\in\Gamma_{RC},\ \zeta\ni(z,z)}wt(\zeta) \le 2m n^4\pi_{RC}(z).
      \end{align*}
  \end{enumerate}
  \label{lem:congestion:gammahat}
\end{lemma}
\begin{proof}
  Fix $\gamma$, let $Z$ be a random path lifted from $\gamma$ 
  and $\ell$ be the length of $\gamma$.  Thus the path is $\gamma=(w_1,\ldots,w_\ell)$
  and, in particular, the final state of the path is~$w_\ell$.
  For a state $w\in\gamma$, let $i(\gamma,w)$ be index of $w$ in $\gamma$ and $k(w,e)$ be the index of $e$ in $|E\backslash w|$.
  Any $w$ only appears once in $\gamma\in\Gamma_{worm}$ and hence $i(\gamma,w)$ is well defined.

  We want to bound the traffic in $\Gamma_{RC}$ that goes through $(z,z')$.  
  Depending on $z'$, we have three cases.
  \begin{enumerate}
    \item First assume that $z'=z\cup\{e\}$ where $e\not\in z$.
      The traffic may be from $\Gamma_{RC}'$ transitions or from the part we append at the end of each $\Gamma_{RC}'$ path.
      Hence we have the following bound:
      \begin{align*}
        \sum_{\zeta\in\Gamma_{RC},\ \zeta\ni(z,z')}wt(\zeta)&=
        \sum_{w\subseteq z}\Bigg(\sum_{\gamma\ni(w,w\cup\{e\})}wt(\gamma)\textnormal{Pr}_{\gamma}\left(Z_{i(\gamma,w)} = z, Z_{i(\gamma,w)+1} = z'\right)\\
        &\hspace{1.5cm} + \sum_{\gamma=(w_1,\ldots,w_\ell),\ w_\ell=w}wt(\gamma)\textnormal{Pr}_{\gamma}\left(Z_{\ell+k(w,e)-1} = z, Z_{\ell+k(w,e)} = z'\right)\Bigg)\\
        &= \sum_{w\subseteq z}\Bigg(\sum_{\gamma\ni(w,w\cup\{e\})}wt(\gamma)\textnormal{Pr}_{\gamma}\left(Z_{i(\gamma,w)}= z\right) \\
        &\hspace{1.5cm} + \sum_{\gamma,\ w_\ell=w}wt(\gamma)\textnormal{Pr}_{\gamma}\left(Z_{\ell+k(w,e)-1} = z\right)p'\Bigg) \\
        &= \sum_{w\subseteq z}\delta(w,z)\left(\sum_{\gamma\ni(w,w\cup\{e\})}wt(\gamma)+ \sum_{\gamma,\ w_\ell=w}wt(\gamma)p'\right),
      \end{align*}
      where $p'=\frac{p}{1-p}$.
      Hence by Theorem \ref{thm:path:worm},
      \begin{align*}
        \sum_{\zeta\in\Gamma_{RC},\ \zeta\ni(z,z')}wt(\zeta) 
        & = \sum_{w\subseteq z}\delta(w,z)\left(\sum_{\gamma\ni(w,w\cup\{e\})}wt(\gamma)+ \sum_{\gamma,\ w_\ell=w}wt(\gamma)p'\right)\\
        & \le \sum_{w\subseteq z}\delta(w,z)\left(n^4\pi_{worm}(w)\frac{p}{1-p}+ \pi_{even}(w)p'\right) \\        
        & = p'n^4\sum_{w\subseteq z}\delta(w,z)\pi_{worm}(w) + p'\sum_{w\subseteq z, w\in\Omega_0}\delta(w,z)\pi_{even}(w)\\
        & = p'n^4\widehat{\pi}(z) + p'\pi_{RC}(z)\\
        & \le 2 p' n^4 \pi_{RC}(z),
      \end{align*}
      where we use Lemma \ref{lem:worm:coupling} in the last line.
      Also note that $\pi_{even}(w)=0$ if $w\not\in \Omega_0$.
    
\item Next assume that $z'=z\backslash\{e\}$ where $e\in z$.
      Similar to the previous case, we have that
      \begin{align*}
        \sum_{\zeta\in\Gamma_{RC},\ \zeta\ni(z,z')}wt(\zeta)&=
        \sum_{w\subseteq z,\, w\ni e}\,\,\sum_{\gamma\ni(w,w\backslash\{e\})}wt(\gamma)\textnormal{Pr}_{\gamma}\left(Z_{i(\gamma,w)} = z, Z_{i(\gamma,w)+1} = z'\right)\\
        &\hspace{1.5cm} + \sum_{w\subseteq z,\, w\not\ni e}\,\,\sum_{\gamma,\ w_\ell=w}wt(\gamma)\textnormal{Pr}_{\gamma}\left(Z_{\ell+k(w,e)-1} = z, Z_{\ell+k(w,e)} = z'\right)\\
        &= \sum_{w\subseteq z,\, w\ni e}\,\,\sum_{\gamma\ni(w,w\backslash\{e\})}wt(\gamma)\textnormal{Pr}_{\gamma}\left(Z_{i(\gamma,w)}= z\right)(1-p')\\
        &\hspace{1.5cm}+ \sum_{w\subseteq z,\, w\not\ni e}\,\,\sum_{\gamma,\ w_\ell=w}wt(\gamma)\textnormal{Pr}_{\gamma}\left(Z_{\ell+k(w,e)-1} = z\right)(1-p') \\
        = &\sum_{w\subseteq z,\, w\ni e}(1-p')\delta(w,z)\,\,\sum_{\gamma\ni(w,w\backslash\{e\})}wt(\gamma)\\
        &\hspace{1.5cm}+ \sum_{w\subseteq z,\, w\not\ni e}(1-p')\delta(w,z)\,\,\sum_{\gamma,\ w_\ell=w}wt(\gamma).
      \end{align*}
      Again we use Theorem \ref{thm:path:worm} and Lemma \ref{lem:worm:coupling}:
      \begin{align*}
        \sum_{\zeta\in\Gamma_{RC},\ \zeta\ni(z,z')}wt(\zeta) 
        & \le \sum_{w\subseteq z}\delta(w,z)(1-p')\left(n^4\pi_{worm}(w)+ \pi_{even}(w)\right) \\
        & \le (1-p')n^4\sum_{w\subseteq z}\delta(w,z)\pi_{worm}(w) + (1-p')n^4\sum_{w\subseteq z, w\in\Omega_0}\delta(w,z)\pi_{even}(w)\\
        & = (1-p')n^4\widehat{\pi}(z)+(1-p')\pi_{RC}(z)\\
        & \le 2 (1-p') n^4 \pi_{RC}(z).
      \end{align*}
    \item At last we handle the case that $z=z'$.
      Then we have the following bound
      \begin{align*}
        \sum_{\zeta\in\Gamma_{RC},\ \zeta\ni(z,z)}wt(\zeta)=
        & \sum_{w\subseteq z}\Bigg(\sum_{\gamma\ni w}wt(\gamma)\textnormal{Pr}_{\gamma}\left(Z_{i(\gamma,w)}= z, Z_{i(\gamma,w)+1}=z\right)\\
        & \hspace{1cm} + \sum_{\gamma,\ w_\ell=w}wt(\gamma)\sum_{i=1}^{|E\backslash w|}
        \textnormal{Pr}_{\gamma}\left(Z_{\ell(\gamma)+i-1}= z, Z_{\ell(\gamma)+i}=z\right)\Bigg)\\
        \le & \sum_{w\subseteq z}\left(\sum_{\gamma\ni w}wt(\gamma)\textnormal{Pr}_{\gamma}\left(Z_{i(\gamma,w)}= z\right)
        + \sum_{\gamma,\ w_\ell=w}wt(\gamma)\delta(w,z) |E\backslash w| \right) \\
        = & \sum_{w\subseteq z}\delta(w,z)\left(\sum_{e\in z}\sum_{\gamma\ni(w,w\oplus\{e\})}wt(\gamma)+ |E\backslash w|\sum_{\gamma,\ w_\ell=w}wt(\gamma)\right).\\
      \end{align*}
      By Theorem \ref{thm:path:worm} and Lemma \ref{lem:worm:coupling},
      \begin{align*}
        \sum_{\zeta\in\Gamma_{RC},\ \zeta\ni(z,z)}wt(\zeta) 
        & \le m \sum_{w\subseteq z}\delta(w,z)\left(n^4\pi_{worm}(w)+ \pi_{even}(w)\right)\\
        & = 2m n^4\pi_{RC}(z). \qedhere
      \end{align*} 
  \end{enumerate}
\end{proof}


\begin{lemma}
  $\rho(\Gamma_{RC})\le 8 m^2n^4$.
  \label{lem:remove-hat}
\end{lemma}
\begin{proof}  
  For any transition $(z,z')$ where $z'=z\cup\{e\}$ for some $e\not\in z$, 
  \begin{align*}
    \frac{L(\Gamma_{RC})}{\pi_{RC}(z)P_{RC}(z,z')}\sum_{\zeta\in\Gamma_{RC},\ \zeta\ni(z,z')}wt(\zeta)
    & \le \frac{L(\Gamma_{RC})}{\pi_{RC}(z)P_{RC}(z,z')}\cdot\frac{p}{1-p}\cdot 2n^4\pi_{RC}(z)\\
    & \le 2 m n^4\cdot\frac{p}{1-p}\cdot\frac{2m}{\min\{1,\frac{2p}{2(1-2p)}\}}\\
    & \le 4 m^2n^4,
  \end{align*}
  where we use Lemma \ref{lem:congestion:gammahat} in the first line and $p\le 1/2$ in the last.

  Similarly, for a transition $(z,z')$ where $z'=z\backslash\{e\}$ for some $e\in z$, 
  \begin{align*}
    \frac{L(\Gamma_{RC})}{\pi_{RC}(z)P_{RC}(z,z')}\sum_{\zeta\in\Gamma_{RC},\ \zeta\ni(z,z')}wt_{RC}(\zeta)
    & \le \frac{L(\Gamma_{RC})}{\pi_{RC}(z)P_{RC}(z,z')}\cdot\frac{1-2p}{1-p}\cdot 2 n^4\pi_{RC}(z)\\
    & \le 2 m n^4\cdot\frac{1-2p}{1-p}\cdot\frac{2m}{\min\{1,\frac{1-2p}{2p}\}}\\
    & \le 8 m^2n^4,
  \end{align*}
  where we use Lemma \ref{lem:congestion:gammahat} in the first line and $p\le 1/2$ in the last.

  For any transition $(z,z')$ where $z'= z$, since the chain is lazy, $P_{RC}(z,z')\ge 1/2$ and
  \begin{align*}
    \frac{L(\Gamma_{RC})}{\pi_{RC}(z)P_{RC}(z,z')}\sum_{\zeta\in\Gamma_{RC},\ \zeta\ni(z,z')}wt_{RC}(\zeta)
    & \le \frac{L(\Gamma_{RC})}{\pi_{RC}(z)P_{RC}(z,z')} \cdot 2 m n^4 \pi_{RC}(z)\\
    & \le 4 m^2 n^4,
  \end{align*}
  where we use Lemma \ref{lem:congestion:gammahat} in the first line.
\end{proof}

Combining Theorem \ref{thm:mixingtime:congestion} and Lemma \ref{lem:remove-hat} gives us desired mixing time bound for $P_{RC}$.

\begin{theorem}
  For the random cluster model with parameters $0<p<1$ and $q=2$,
  \begin{align*}
    \tau_{\epsilon}(P_{RC})\le 8 n^{4} m^2 (m\ln (1-p)^{-1}+\ln\epsilon^{-1}).
  \end{align*}
  \label{thm:RC:fast}
\end{theorem}


\bibliographystyle{plain}
\bibliography{IsingRC14}

\appendix

\section{Equivalence of the three views}
\label{sec:equivalence}
The equivalence between the Ising model and the random cluster model with $q=2$ can be found in \cite{Grimmett:book}.
It can also be seen as the following.
Instead of assigning vertices $0$ or $1$, we may assign ``equal'' or ``independent'' to edges.
Each ``equal'' edge has an weight of $\beta-1$, and ``independent'' has weight $1$.
This is the same as the Ising model, since for each edge, 
if the two end points are equal, the weight is $\beta-1+1=\beta$,
whereas if the two endpoints are not equal, the weight is $1$.
For a subset $S\subseteq E$ of edges being ``equal'', 
each component of $S$ has two possible assignments.
Therefore the weight of $S$ is $(\beta-1)^{|S|}2^{\kappa(S)}$.
After rescaling by $\beta^{|E|}$ this matches the random cluster formulation \eqref{eqn:RC} with $p=1-\frac{1}{\beta}$ and $q=2$.
This gives the first equality of \eqref{eqn:equivalence}.

The equivalence between the Ising model and even subgraphs model can be explained via a holographic transformation 
by Hadamard matrix $H=\trans{1}{1}{1}{-1}$.\footnote{For basics of holographic transformations, see e.g.\ \cite{CGW13}.}
In the Ising model, vertices have functions \textsc{Equality} on its adjacent $d$ many half-edges,
which after the transformation becomes \textsc{Even} function, defined as follows:
\begin{align*}
  \textsc{Even}(x_1,\cdots,x_d)=
  \begin{cases}
    2 & \textrm{if }\bigoplus_i x_i=0;\\
    0 & \textrm{otherwise}.\\
  \end{cases}
\end{align*}
On the edges, the function (on the two half-edges) is 
\begin{align*}
  \textsc{Ising}(x_1,x_2)=
  \begin{cases}
    \beta & \textrm{if }x_1=x_2;\\
    1 & \textrm{otherwise},\\
  \end{cases}
\end{align*}
whereas after the transformation it is
\begin{align*}
  \textsc{WEQ}(x_1,x_2)=
  \begin{cases}
    \frac{\beta+1}{2} & \textrm{if }x_1=x_2=0;\\
    \frac{\beta-1}{2} & \textrm{if }x_1=x_2=1;\\    
    0 & \textrm{otherwise},\\
  \end{cases}
\end{align*}
a weighted equality function.
Therefore, for a subset $S$ of edges (both half-edges are $1$), its weight is 
\begin{align*}
  wt(S)=
  \begin{cases}
    2^{|V|}\left( \frac{\beta-1}{2} \right)^{|S|}\left( \frac{\beta+1}{2} \right)^{|E\backslash S|}
    & \textrm{if }S\in\Omega_{even};\\
    0 & \textrm{otherwise}.
  \end{cases}
\end{align*}
The requirement of $S\in\Omega_{even}$ arises because each vertex requires even degree,
and when all degree constraints are satisfied, the vertices contribute $2^{|V|}$ in total.
We may rewrite the weight of $S\in\Omega_{even}$:
\begin{align*}
  2^{|V|}\left( \frac{\beta-1}{2} \right)^{|S|}\left( \frac{\beta+1}{2} \right)^{|E\backslash S|}
  = 2^{|V|}\beta^{|E|} \left( \frac{1}{2}\left( 1-\frac{1}{\beta} \right) \right)^{|S|}\left( \frac{1}{2}\left( 1+\frac{1}{\beta} \right) \right)^{|E\backslash S|}
\end{align*}
Hence setting $p=\frac{1}{2}\left( 1-\frac{1}{\beta} \right)$ matches \eqref{eqn:even} 
and taking out appropriate scaling factor yields the second equality of \eqref{eqn:equivalence}.

\section{Congestion of the worm process}
\label{sec:worm}

Throughout this section fix $p\le 1/2$.
Recall that $\Omega_k$ is the collection of subgraphs where $k$ many vertices have odd degrees.
Then $\Omega_0=\Omega_{even}$, and $\Omega_0\cup\Omega_2=\Omega_{worm}$.
Define
\begin{align*}
  Z_k := \sum_{S\in\Omega_k} w_p(S),
\end{align*}
where $w_p(S)=p^{|S|}(1-p)^{|E\backslash S|}$.
Then $Z_0=Z_{even}(p)$ and $Z_{worm}(p)=Z_0+n^{-2}Z_2$.

If we adopt the holographic transformation view of the even subgraphs model,
then a vertex that only allows odd degrees is equivalent to the following function:
\begin{align*}
  \textsc{Odd}(x_1,\cdots,x_d)=
  \begin{cases}
    2 & \textrm{if }\bigoplus_i x_i=1;\\
    0 & \textrm{otherwise}.\\
  \end{cases}
\end{align*}
Transforming back to the Ising model,
this vertex is still an \textsc{Equality} on all adjacent half-edges,
but with a weight of $-1$ when all half-edges are assigned $1$.
Hence for every $u,v\in V$,
\begin{align}
  Z_{u,v}:=\sum_{S\in\Omega(u,v)}w_p(S)\le Z_0,
  \label{eqn:two:holes}
\end{align}
because the left hand side can be transformed to the original Ising with $u$ and $v$ having weights~$-1$.
We can sum over all possible pairs of vertices in \eqref{eqn:two:holes}, getting $Z_2\le {n\choose 2} Z_0$.

\begin{lemma} 
  $Z_2\le{n\choose 2}Z_0$.
  \label{lem:Z4}
\end{lemma}

Lemma \ref{lem:Z4} implies that $Z_{worm}=Z_0+n^{-2}Z_2\le Z_0+n^{-2}{n\choose 2}Z_0\le 3Z_0/2$.

Now we are ready to prove Theorem \ref{thm:path:worm}.

\begin{proof}[Proof of Theorem \ref{thm:path:worm}]
  Let $I$ and $F$ be two configurations in $\Omega_0$, denoting the initial and final states.
  Then $I\oplus F\in\Omega_0$.
  The canonical path from $I$ to $F$ will be identical to those in \cite{JSising}.
  Fix an arbitrary ordering of all cycles in $G$.
  For each cycle we designate a starting vertex and a direction around the cycle.
  Hence each cycle is an ordered tuple of edges.
  Since $I\oplus F$ is an even subgraph, we can cover $I\oplus F$ by a collection of edge-disjoint cycles.
  Let $\{C_1,\cdots,C_r\}$ be the first such in our ordering.
  Let $e_1,\cdots,e_k$ be the edges of $\{C_1,\cdots,C_r\}$ taken in order (first order the 
  edges according to the cycle they occur in, and then by their position within the cycle,
  counting from the start vertex).
  The canonical path $\gamma$ from $I$ to $F$ is defined to be $Z_0=I$, $Z_{i}=Z_{i-1}\oplus e_i$, and $Z_k=F$.
  Intuitively the canonical path unwinds $C_i$ one by one from $i=1$ to $i=r$.
  Clearly $L=L(\Gamma_{worm})\le m$ as it can use every edge at most once.

  This path is always in $\Omega_0\cup\Omega_2$ because 
  if we start to unwind a cycle, then the current state is an even subgraph.
  If we are unwinding a path, then we always flip an edge that is adjacent to an odd degree vertex.

  For any transition $(w,w')$ where $w'=w\oplus e$ for some edge $e\in E$,
  we use a combinatorial encoding as in \cite{JSising} for all paths passing through $(w,w')$.
  For any two configurations $I,F\in\Omega_0$,
  let $\varphi(I,F)=I\oplus F\oplus w$.
  We claim that $\varphi:\Omega_0^2\rightarrow\Omega_0\cup\Omega_2$ is an injection.
  This is because given $(w,w')$ and $U=\varphi(I,F)$, we can recover the unique $(I,F)$.
  First, since $w\oplus U=I\oplus F$, all edges \emph{not} in $w\oplus U$ have the same state in both $I$ and $F$,
  and their states are the same as those in $w$.
  Then for edges in $w\oplus U$, due to the construction of the canonical path,
  there is a unique ordering among those edges, including $e=w\oplus w'$.
  For any edge before $e$, its status in $w$ has been changed to that in $F$, and its status in $U$ is still the same as that in $I$.
  For any edge after $e$ (including $e$ itself), its status in $w$ is still the same as that in $I$, and in $U$ is the same as in $F$.

  Recall that $w_p(S)=p^{|S|}(1-p)^{|E\backslash S|}$ for any subset of edges $S\subset E$.
  Since $I\oplus F=w\oplus U$ and $I\cap F=w\cap U$, we have that
  \begin{align*}
    w_p(I)w_p(F) = w_p(w)w_p(U).
  \end{align*}
  Therefore,
  \begin{align*}
    \sum_{\gamma\ni(w,w')}wt(\gamma) & = \sum_{\substack{I,F\in\Omega_0^2\\ \gamma_{IF}\ni(w,w')}}\pi_{even}(I)\pi_{even}(F)
    = \sum_{\substack{I,F\in\Omega_{0}^2\\ \gamma_{IF}\ni(w,w')}}\frac{w_p(I)w_p(F)}{Z_0^2}\\
    & = \sum_{\substack{I,F\in\Omega_{0}^2\\ \gamma_{IF}\ni(w,w')}}\frac{w_p(w)w_p(\varphi(I,F))}{Z_0^2}\\
    & \le w_p(w)\sum_{U\in \Omega_0\cup\Omega_2}\frac{w_p(U)}{Z_0^2}\\
    & = \frac{Z_0+Z_2}{Z_0^2} \cdot w_p(w).
  \end{align*}
  By the definition of $\pi_{worm}$ \eqref{eqn:worm:weight} and \eqref{eqn:worm:measure}, 
  $\pi_{worm}(w)=\frac{w_{worm}(w)}{Z_{worm}} \ge \frac{w_p(w)}{n^2Z_{worm}}$.
  This implies that 
  \begin{align*}
    \sum_{\gamma\ni(w,w')}wt(\gamma)
    & \le \frac{Z_0+Z_2}{Z_0}\cdot\frac{Z_{worm}}{Z_0}\cdot n^2 \pi_{worm}(w)\\
    & \le \left( 1+{n\choose 2} \right)\left( 1+{n\choose 2}n^{-2} \right)\tag{by Lemma \ref{lem:Z4}}n^2 \pi_{worm}(w)\\
    & \le n^4 \pi_{worm}(w).
  \end{align*}      
  
  For the last claim of the theorem,
  let $w'=w\cup\{e\}$ for some $e\not\in w$.
  We can do the same combinatorial encoding for $w'$.
  That is, let $U'=\varphi'(I,F)=I\oplus F\oplus w'$.
  It is easy to verify as above that $\varphi'$ is an injection.
  Then as above,
  \begin{align*}
    \sum_{\gamma\ni(w,w')}wt(\gamma)
    & \le \frac{Z_0+Z_2}{Z_0^2} \cdot w_p(w')\\
    & = \frac{Z_0+Z_2}{Z_0^2} \cdot w_p(w) \cdot \frac{p}{1-p}\\
    & \le n^4 \pi_{worm}(w)\frac{p}{1-p}.\qedhere
  \end{align*}  
\end{proof}
\end{document}